\documentclass[a4paper]{article}
\usepackage{amsmath}
\usepackage{amssymb}
\usepackage{graphicx}
\usepackage{natbib}
\usepackage{bm}
\usepackage{url}
\usepackage{amsthm}
\usepackage{thmtools}
\usepackage[margin=2cm,nohead]{geometry}
\usepackage{setspace}
\usepackage{paralist}
\usepackage{lineno}
\usepackage{appendix}

\declaretheorem[style=definition]{definition}
\declaretheorem[style=plain]{theorem}
\declaretheorem[style=plain]{lemma}
\begin{document}
\title{Equivalence relations on ecosystems}
\author{Matthew Spencer\\School of Environmental Sciences, University of Liverpool, Liverpool, L69, 3GP, UK.\\m.spencer@liverpool.ac.uk}
\maketitle

\section*{Abstract}
In abstract terms, ecosystem ecology is about determining when two ecosystems, superficially different, are alike in some deeper way. An external observer can choose any ecosystem property as being important. In contrast, two ecosystems are equivalent from the point of view of the organisms they contain if and only if for each species, the proportional population growth rate does not differ between the ecosystems. Comparative studies of ecosystems should therefore focus on patterns in proportional population growth rates, rather than patterns in other properties such as relative abundances. Popular activities such as measuring dissimilarity, and representing dissimilarity via ordination, can then be done from the point of view of the organisms in ecosystems. Summarizing the state of an ecosystem under this approach remains challenging. In general, the dynamics on equivalence classes of ecosystems defined in this way are structurally different from the dynamics of ecosystems as seen by an external observer. This may limit the extent to which natural selection can act on ecosystem structure. 

\section{Introduction}

In abstract terms, ecosystem ecology is about determining when two ecosystems, superficially different, are alike in some deeper way. Previous authors have chosen properties such as abundances \citep[p. 7]{Ginzburg83}, relative abundances \citep[p. 328]{Legendre12}, diversity \citep{Jost06} and complexity of energy flow pathways \citep{Ulanowicz86} as the relevant ways in which ecosystems may be alike. Since there are many such properties, it is worth thinking about which ones to choose. In order to make this choice in a principled way, it is necessary to distinguish between an external observer of the ecosystem, who can decide which properties matter to them, and an organism in the ecosystem, for whom the choice of properties is fixed.

The properties that matter to an organism are closely related to its Hutchinson niche. \citet{Hutchinson57} defined the niche of an organism as the set of states of the environment permitting a species to persist indefinitely. Later work makes it clear that by ``persist indefinitely'', it was meant that the proportional population growth rate for the species was non-negative \citep[p. 194]{Hutchinson78}, where proportional population growth rate is the per-capita population growth rate when abundance is measured in individuals, the mass-specific population growth rate when measured as biomass, and so on. This initial view of niche space was essentially static. \citet{Maguire73} introduced both structure and dynamics into niche space. Structure was provided by level sets (contours, if niche space is two-dimensional) of equal proportional population growth rate, and dynamics by movement through niche space, driven either by external changes, or as a consequence of population growth. Maguire explicitly stated that this view of niche space allows us to examine ``the total environment of a species, a population, or an individual \ldots through its `biological eyes'\thinspace'', in other words as an organism within the ecosystem would see it, rather than as an external observer. Of course, organisms do not ``see'' population growth rate, so that the ``biological eyes'' of a species must be interpreted as the outcome of the process relating population growth rate to environment.

Exponential growth of a population occurs when ``nothing happens in the environment'' \citep{Ginzburg86}. However, it is worth thinking about what ``nothing'' means. \citet{Maguire73}, and later \citet{Tilman80}, concentrated on movement through niche space, from one level set to another (Figure \ref{fig:contours}, solid arrow). Since proportional population growth rate is changing, population growth is not exponential, and something is happening to the environment. A situation not considered in detail by \citet{Maguire73} and \citet{Tilman80} is shown by the dashed arrow in Figure \ref{fig:contours}. Here, the environment is changing, but in such a way that we remain in the same level set in niche space. Although this is unlikely to occur in nature, it is important conceptually. In such cases, ``nothing happens'' from the point of view of the species, even though to an external observer, something is happening. In general, two ecosystems which are superficially different can be equivalent from the point of view of a species if they are in the same level set in niche space and therefore lead to the same proportional population growth rate. In other words, proportional population growth rates are the natural units in which to measure ecosystem state from the point of view of an organism, and there is no necessary connection between these natural units and the properties of the system that seem most natural to an external observer \citep[p. xii]{Rosen78}.

This natural definition of ecosystem state suggests some immediate questions. First, applying this definition across all species in an ecosystem leads to a high-dimensional state. Ecosystem properties visible to an external observer, such as diversity, are often summarized in a low-dimensional way, for example using diversity indices \citep{Jost06}. Can the same be done for the natural measure of ecosystem state? Second, the relationship between the dynamics of an ecosystem (including all properties that an external observer could measure) and the dynamics of ecosystem states (as experienced by organisms) may have important consequences for attempts to explain patterns in ecosystem structure \citep{Borrelli15}. Natural selection cannot distinguish between groups of organisms with the same proportional population growth rates. As a result, there will be variation among ecosystems (visible to an external observer) on which natural selection cannot act. To what extent does this limit the role of natural selection as an explanation for ecosystem structure?

Here, I define equivalence of ecosystems from the point of view of the organisms involved, in terms of proportional population growth rates. I outline the relationship between dynamics on equivalence classes of ecosystems (from the point of view of organisms) and dynamics as seen by an external observer. I identify classes of ecosystems differing in the relationship between these two kinds of dynamics. I discuss the consequences of these ideas for comparative studies of ecosystems, summaries of ecosystem state, and the mechanisms that may generate regularities at the ecosystem level.

\section{Equivalence from the point of view of organisms}
\begin{definition}
Let $\Omega$ be a volume in two- or three-dimensional physical space. Let $\mathbf x = (x_1, x_2, \ldots, x_n)$ be the abundances (e.g. numbers of individuals, if individuals are well-defined, or cover or biomass otherwise) of all the $n$ species present in $\Omega$ ($x_i \in \mathbb R_{>0}, i = 1, \ldots, n$). Let $\mathbf y = {y_1, y_2, \ldots, y_m}$ be the values of all the physicochemical variables affecting any of these species ($y_i \in \mathbb R, i = 1, \ldots, m$). Then $s = \{\Omega, \mathbf x, \mathbf y \}$ is an \emph{ecosystem}.
\label{def:ecosystem}
\end{definition}

Definition \ref{def:ecosystem} is not greatly different from standard usage, but it is necessary to have a precise definition. Let $S$ be the set of ecosystems in which exactly the same set of $n$ species are present. Let $\alpha$ be an endomap of $S$ \citep[a function with domain and codomain $S$:][p. 15]{Lawvere09}, describing ecosystem dynamics within $S$. In what follows, I assume for simplicity that dynamics operate in discrete time. Essentially the same arguments as those below can be made in continuous time, the only difference being that there must then be an endomap $\alpha_t$ for each real number $t$, satisfying $\alpha_0 = 1_S$ (the identity in $S$) and $\alpha_{t + u} = \alpha_t \circ \alpha_u$, i.e. the composition $\alpha_t$ following $\alpha_u$ \citep[p. 169]{Lawvere09}.

Let $r_i: S \to \mathbb R$ be a function such that $r_i(s)$ is the contribution of endogenous processes (e.g. births and deaths) to the proportional growth rate of the $i$th species. In a finite population, this is interpreted as the expected value over demographic stochasticity. There is no need to consider environmental stochasticity, because by definition, all the variables that affect $r_i$ are specified in $s$. I do not require that $s$ is a closed ecosystem, but I do not include immigration and emigration in $r_i(s)$. This is consistent with the view that immigration and emigration should not be considered when determining the suitability of an environment for a species, which resolves some of the problems with connecting the definition of a niche to the distribution of a species \citep{Drake17}. 

To the $i$th species, two ecosystems $s, s' \in S$ are equivalent if and only if $r_i(s) = r_i(s')$. As argued above, when this condition is satisfied, the two ecosystems lie in the same level set in niche space for the $i$th species, so that from the point of view of the species, ``nothing happens'' if we move from one to the other, even though the ecosystems may appear different to an external observer. In other words, a unit of abundance of the $i$th species in ecosystem $s$ would neither benefit nor suffer in evolutionary terms if exchanged with a unit of abundance of $i$ from ecosystem $s'$. Let $\sim_i$ be the relation defined on $S$ by $s \sim_i s'$ if and only if $r_i(s) = r_i(s')$. This is an equivalence relation because it is reflexive, symmetric and transitive \citep[p. 28]{Halmos74}. The elements of the quotient set $S / {\sim}_i$ (the equivalence classes of $\sim_i$ in $S$) are the level sets in niche space for species $i$, provided that the abundance of any species having a direct effect on $r_i$ is included as a niche axis \citep{Maguire73}.

\begin{definition}
Let $r$ be the function
\begin{equation*}
\begin{aligned}
r: S &\to \mathbb R^n \\
s &\mapsto (r_1(s), \ldots , r_n(s)),
\end{aligned}
\end{equation*}
which maps ecosystems to $n$-tuples of real numbers representing proportional population growth rates of all species. Because the set of such $n$-tuples is important, it is worth giving it a symbol ($R$) and a name: the growth space of the ecosystem \citep{Spencer15}. Let $\sim$ be the relation on $S$ defined by $r$ (i.e. $s \sim s'$ means that $r(s) = r(s')$). Again, this is reflexive, transitive and symmetric, so it is an equivalence relation. Then I say that ecosystems $s, s' \in S$ are \emph{equivalent} (from the point of view of every species) if and only if $s \sim s'$. In other words, two ecosystems are equivalent if and only if for each species, the proportional population growth rate does not differ between the ecosystems.
\end{definition}

The elements of the quotient set $S / {\sim}$ are the intersections of the quotient sets $S / {\sim}_1, \ldots, S / {\sim}_n$, i.e. $S / {\sim} = S / \left( \cap_{i=1}^n {\sim}_i \right)$. In biological terms, these are the intersections of a given set of level sets for each species in niche space. Note that some of these intersections may be empty. Studying intersections of sets in niche space has been productive in the past. For example, \citet{Hutchinson57} focused on intersections of the sets $r_i \geq 0$, and \citet{Tilman80} focused on intersections of the level sets $r_i = 0$, in order to study the potential for coexistence. However, the intersections of other level sets are also biologically important, a point I return to in the discussion.

I do not assume that either $\alpha$ or $r$ has any particular form. In general, the equations describing ecosystem dynamics are unknown. For example, the Lotka-Volterra equations can usefully be thought of as a second-order Taylor polynomial approximation to some more complicated system \citep[p. 117]{Hutchinson78}, but there are few situations in which one would believe that these are the true equations. It is possible to constrain the form of the functions $r_i$ based on a few axioms \citep{Cropp15}. Although I do not follow this up here, it may lead to a deeper understanding of the range of possible dynamics on equivalence classes.

I also do not assume that specifying $r$ is sufficient to specify $\alpha$. Although endogenous dynamics are important, immigration and emigration of organisms, and external factors influencing environmental conditions, must also be specified in order to know the future state of an ecosystem. Closed ecosystems have received more theoretical attention, but ecosystems with input and output of of nutrients and organisms can have qualitatively different dynamics \citep{Loreau04}.

It is worth emphasizing again that it is proportional growth rates and nothing else that determine whether two ecosystems are equivalent from the point of view of the organisms involved. For example, two ecosystems which happen to have the same abundance of every species (and may therefore be viewed as equivalent by an external observer) may or may not be equivalent to the organisms involved. On the other hand, each internal equilibrium of a deterministic system is a member of $S$, and all these equilibria are equivalent, since $r_i(s) = 0$ for $i = 1, \ldots, n$, by definition. Thus what are usually called alternative stable states are invisible (in evolutionary terms) to the organisms involved, if every species is present in each of these states.

\section{The category of sets with endomaps}

A category can be thought of as a set of objects $A, B, C, \ldots$ and a set of arrows $f, g, h, \ldots$, such that:
\begin{enumerate}
\item Each arrow $f$ has some object $A$ as its domain (source) and some object $B$ as its codomain (target);
\item There is an identity, consisting of an arrow $1_A$ for each object $A$ with domain and codomain $A$;
\item Any pair of arrows $f, g$ such that the codomain of $f$ is the domain of $g$ can be composed to form a composite arrow $g \circ f$ from the domain of $f$ to the codomain of $g$;
\item Composition is associative, i.e. $h \circ (g \circ f) = (h \circ g) \circ f$;
\item Composition satisfies the unit laws, that for arrows $f$ with codomain $B$, and $g$ with domain $B$, $1_B \circ f = f$ and $g \circ 1_B = g$;
\end{enumerate}
\citep[p. 21]{Lawvere09}.

For example, a set of ecosystems $S$ with an endomap $\alpha$ is an object in the category of sets with endomaps \citep[p. 136]{Lawvere09}. An arrow $f$ in this category from a set $X$ with endomap $\gamma$ to a set $Y$ with endomap $\delta$ must preserve the structure of the endomap, in the sense that it must satisfy
\begin{equation}
  f \circ \gamma = \delta \circ f.
  \label{eq:map}
\end{equation}

Intuitively, this means that we can either follow dynamics on $X$ and then map the result to $Y$, or map to $Y$ and then follow the corresponding dynamics of the result on $Y$. Thus the dynamical structure on $X$ defined by the endomap $\gamma$ is preserved in the structure on $Y$ defined by the endomap $\delta$.

\section{Can dynamics on $R$ have the same structure as dynamics on $S$?}
Dynamics on $S$ (visible to an external observer) induce dynamics on $R$ (as experienced by organisms in the ecosystem). We want to know whether these dynamics have the same structure, in the sense of Equation \ref{eq:map}. To answer this question, we must first attempt to specify an endomap $\beta$ on $R$ that describes these dynamics.

A natural choice for $\beta$ is a map taking $r(s)$ to $(r \circ \alpha)(s)$ (the outcome of dynamics on ecosystems, mapped to $R$), if such a map exists. Thus, suppose that $z \in r(S)$. To get from $z = r(s)$ to $\beta(z) = (r \circ \alpha)(s)$, we have to first go back to $S$, then apply $\alpha$ and finally go from the result of this to $R$. The function $r$ is not in general one-one, so it will not in general have a retraction $\tilde{r}$ that undoes it in the sense that $\tilde{r} \circ r = 1_S$ \citep[p. 53]{Lawvere09}. However, we can construct the function
\begin{equation*}
  \begin{aligned}
    r': r(S) &\to S \\
    z &\mapsto s^*, \\
  \end{aligned}
\end{equation*}
where $s^*$ is an arbitrary representative of the set $\{s \in S: r(s) = z\}$. Then if the function $r \circ \alpha \circ r'$ exists, it is the natural choice for $\beta$ on $r(S)$. For elements of $R$ outside the image set of $S$ under $r$, we can define $\beta$ in an arbitrary way, say $\beta = 1_R$.

It is clear that we will not always be able to construct $\beta$ in this way. In fact, if we cannot, then there is no endomap on $R$ such that $r$ is a structure-preserving map from $S$ to $R$.

\begin{theorem}
The map $r: S \to R$ can be a structure-preserving map if and only if the endomap $\alpha$ on $S$ satisfies the condition that
\begin{equation}
s \sim s' \implies \alpha(s) \sim \alpha(s') \quad \forall s, s' \in S.
\label{eq:structcond}
\end{equation}
\label{theorem:structcond}
\end{theorem}
\begin{proof}
First, I show that if Condition \ref{eq:structcond} holds, then the endomap $\beta$ on $R$ is structure-preserving. If the condition holds, then by the definition of $\sim$, $s \sim s' \implies (r \circ \alpha)(s) = (r \circ \alpha)(s')$. Then
\begin{equation}
\begin{aligned}
\beta: R &\to R \\
z & \mapsto \begin{cases} (r \circ \alpha)(s) & \text{if}\ z \in r(S),\\
  z & \text{otherwise}
  \end{cases}
\end{aligned}
\label{eq:naturalRmap}
\end{equation}
is a valid endomap on $R$ (because it has domain and codomain $R$, and associates a single element of its codomain with each element of its domain). It also satisfies $r \circ \alpha = \beta \circ r$, and is therefore structure-preserving.

Now, I show that if Condition \ref{eq:structcond} does not hold, then there cannot be any endomap on $R$ such that $r$ is structure-preserving. Suppose that there exist $s, s' \in S$ such that $s \sim s'$, but $\alpha(s) \not \sim \alpha(s')$. Then by the definition of $\sim$, $r(s) = r(s')$, but $(r \circ \alpha)(s) \neq (r \circ \alpha)(s')$. There cannot be any function $\gamma$ that maps $r(s) = r(s')$ to both $(r \circ \alpha)(s)$ and $(r \circ \alpha)(s') \neq (r \circ \alpha)(s)$ when these elements are distinct, and hence it is not possible to satisfy $r \circ \alpha = \gamma \circ r$.

I have shown that if Condition \ref{eq:structcond} holds, then there is a natural choice of endomap $\beta$ such that $r$ is a structure-preserving map from $S$ to $R$, and that if it does not hold, then there can be no such map.
\end{proof}

Theorem \ref{theorem:structcond} makes intuitive sense. Condition \ref{eq:structcond} says that for dynamics on the set of equivalence classes to preserve the structure in ecosystem dynamics, the ecosystem dynamics must not separate equivalence classes. For example, in Figure \ref{fig:endomap}a, the structure of $\alpha$ can be preserved by $r$, because $\alpha$ keeps members of equivalence classes together, while in Figure \ref{fig:endomap}b, the structure of $\alpha$ cannot be preserved by $r$ because $s$ and $s'$ are in the same equivalence class but are mapped by $\alpha$ to different equivalence classes. Condition \ref{eq:structcond} is somewhat analogous to the condition under which a function of a Markov chain will be Markovian \citep{Burke58}.

To find examples of endomaps $\alpha$ satisfying Condition \ref{eq:structcond}, I first construct a function $\phi: S \to S$ such that $s \sim s' \iff \phi(s) = \phi(s')$.
\begin{lemma}
Let $\phi$ be the function
\begin{equation*}
\begin{aligned}
\phi: S &\to S \\
s &\mapsto s^* \, \text{such that}\, s \sim s^*,
\end{aligned}
\end{equation*}
i.e. $s^*$ is any fixed representative of the equivalence class of $s$ on $S$. Then $s \sim s' \iff \phi(s) = \phi(s')$.
\end{lemma}
\begin{proof}
If $s \sim s'$, then $\phi(s) = \phi(s') = s^*$. Conversely, if $s \not \sim s'$, then $\phi(s) = s^*$, but $\phi(s') \neq s^*$, since an equivalence relation on $S$ partitions $S$ \citep[p. 28]{Halmos74}, so that $s^*$ cannot be equivalent to both $s$ and $s'$.
\end{proof}
We can now rewrite Condition \ref{eq:structcond} as
\begin{equation}
\phi(s) = \phi(s') \implies (\phi \circ \alpha )(s) = (\phi \circ \alpha )(s') \quad \forall s, s' \in S.
\label{eq:structcond2}
\end{equation}

\section{Classes of ecosystem dynamics}

It is useful to distinguish three classes of ecosystem dynamics, based on whether and how Condition \ref{eq:structcond2} is satisfied:
\begin{enumerate}[(a)]

\item \label{class:simplest} Condition \ref{eq:structcond2} holds because $\phi \circ \alpha = \phi$, so $r$ is a map in the category of sets with endomaps. Some of the possible ways this could occur are:
\begin{enumerate}[(i)]
\item If $\alpha = 1_S$, then $\phi \circ \alpha = \phi$, and Condition \ref{eq:structcond2} holds. This is the trivial case in which ecosystems never change. 

\item Note that $\phi$ is idempotent (i.e. $\phi \circ \phi = \phi$), since $(\phi \circ \phi)(s) = \phi(s^*) = s^* = \phi(s)$, for any $s \in S$. Hence $\alpha = \phi$ also satisfies Condition \ref{eq:structcond2}, and is not equal to $1_S$, provided that at least one equivalence class has more than one member. There is no obvious biological example of this case. 

\item If resource levels change over time, but in such a way that $r$ remains constant (as in Figure \ref{fig:contours}, dashed arrow), then $\alpha \neq 1_S$, but $\phi \circ \alpha = \phi$. This could in principle be achieved in a controlled laboratory system, but does not appear likely in nature.

\item Finally and most importantly, consider an infinite well-mixed space $\Omega$, and a set of species interacting only through resource depletion and production of waste products. The proportional growth rate of each species depends on physicochemical variables $\mathbf y$ and will not in general be zero. Thus in a closed system, we expect abundances $\mathbf x$ to change over time, so $\alpha \neq 1_S$. Furthermore, because $\Omega$ is infinite and well-mixed, $\mathbf y$ does not change over time, so proportional growth rates do not change over time and the abundance of each species grows or declines exponentially. Thus in this case, ecosystems change, while remaining in the same equivalence class, and Condition \ref{eq:structcond2} is satisfied. This Malthusian situation is an important starting point for theory, analogous to the role of a body with no forces acting on it in physics \citep{Ginzburg86}. In the real world, a similar situation can be realized in a chemostat in which the ecosystem is open and proportional population growth rates are constant but not necessarily zero, while abundances in the system do not change.

\end{enumerate}

\item \label{class:endomap} Condition \ref{eq:structcond2} holds even though $\phi \circ \alpha \neq \phi$. In other words, ecosystems change equivalence class over time, but these dynamics keep members of the same equivalence class together, so that $r$ is a map in the category of sets with endomaps. There are several possible examples.

\begin{enumerate}[(i)]
\item \label{class:SDM} Suppose that $r$ depends on $s$ only as a one-one function of a single physicochemical variable $y$, and that changes in abundances $\mathbf x$ do not affect $y$. Then each equivalence class of $S$ contains ecosystems with a single value of $y$, but potentially differing in $\mathbf x$. Changes in $y$ will lead to dynamics among equivalence classes, but the members of an equivalence class will stay together. In idealized stream or soil ecosystems, $y$ could represent detritus, and $\mathbf x$ detritivores with pure donor-controlled dynamics \citep[p. 136]{Pimm_716}, with change over time caused by variation in input and output of the resource. Alternatively, $y$ could be an environmental variable whose effects dominate all other variables, with change over time caused by extrinsic environmental variability.

\item \label{class:freqdep} Suppose that proportional population growth rates in a closed system depend only on $\mathbf x$ through the relative abundances $\bm \rho = \left( \sum_{i = 1}^n x_i \right)^{-1} \mathbf x$. Let $\bm \psi$ be a function from $\mathbb S^{n-1} \times \mathbb R_{\geq 0}$ (where the simplex $\mathbb S^{n-1}$ contains the relative abundances, and $\mathbb R_{\geq 0}$ contains a time difference) to $\mathbb R^n$. Then for some time interval $\Delta t$, ecosystem dynamics $\alpha$ are given by
\begin{equation*}
\begin{aligned}
\alpha: S & \to S, \\
(\Omega, \mathbf x, \mathbf y ) & \mapsto (\Omega, \mathbf x \odot \bm \psi(\bm \rho, \Delta t), \mathbf y ),
\end{aligned}
\end{equation*}
where $\odot$ denotes the elementwise (Hadamard) product. Hence for any given set of relative abundances $\bm \rho$, all ecosystems with abundances of the form $c \bm \rho$ for some positive number $c$ map to ecosystems with abundances of the form $c \mathbf x \odot \bm \psi(\bm \rho, \Delta t)$. Also, for each species $i$,
\begin{equation*}
r_i(s) = \lim_{\Delta t \to 0} \frac{\psi_i(\bm \rho, \Delta t) - 1}{\Delta t},
\end{equation*}
which depends on $\mathbf x$ only through $\bm \rho$. Thus all ecosystems with abundances of the form $c \bm \rho$, for fixed $\bm \rho$, will be in an equivalence class, and will be mapped to the same new equivalence class by $\alpha$. Ecosystems of this kind have purely frequency-dependent dynamics, implicitly assumed in models based only on relative abundances. \citet[section 6.1]{Arditi12} argue that this kind of scaling invariance may be a desirable property. Frequency dependence is certainly possible \citep[e.g.][pp. 134-135]{Hutchinson78}, and is sometimes likely to be important. For example, if space is limiting, and all the available space is always filled, frequency dependence may be the dominant way in which abundances affect proportional population growth rates.

\end{enumerate}

\item \label{class:notendomap} In most cases, Condition \ref{eq:structcond2} will not be satisfied, and so $r$ will not be a map in the category of sets with endomaps. For example, consider a closed ecosystem containing a single species of phytoplankton with abundance $x$, whose proportional population growth rate $(1/x) (\mathrm{d} x / \mathrm{d} t)$ depends on the concentrations of nitrogen ($N$) and phosphorus ($P$), which are used during growth but not recycled. A simple model for such an ecosystem, from \citet{Maguire73}, is
\begin{equation}
\begin{aligned}
\frac{\mathrm{d} x}{\mathrm{d} t} &= x \left( r_{\max} - \sqrt{a (N - N^*)^2 + b (P - P*)^2} \right), \\
\frac{\mathrm{d} N}{\mathrm{d} t} &= - c \frac{\mathrm{d} x}{\mathrm{d} t}, \\
\frac{\mathrm{d} P}{\mathrm{d} t} &= - d \frac{\mathrm{d} x}{\mathrm{d} t}, \\
\end{aligned}
\label{eq:Maguire}
\end{equation}
where $r_{\max}$ is the maximum possible proportional population growth rate, attained at optimum concentrations $N^*, P^*$ of nitrogen and phosphorus respectively, parameters $a$ and $b$ determine how quickly proportional population growth rate declines as nitrogen and phosphorus concentrations, respectively, are moved away from the optimum, and $c, d$ are quantities of nitrogen and phosphorus needed to produce a unit of biomass, respectively. The space $\Omega$ is not explicitly defined, but $s = \{\Omega, x, N, P \}$ is an ecosystem. Consider the map $\alpha$ defined by
\begin{equation*}
\begin{aligned}
\alpha: S &\to S, \\
\{\Omega, x_0, N_0, P_0 \} &\mapsto \{\Omega, x(1), N(1), P(1) \},
\end{aligned}
\end{equation*}
where $x_0, N_0, P_0$ are initial values, and $x(1), N(1), P(1)$ are solutions of Equation \ref{eq:Maguire} after one unit of time. Applying this map to some of the ecosystems in the equivalence class $\{ s \in S: r(s) = 1\}$ (Figure \ref{fig:Maguire}, bold black line) gives sets of ecosystems (Figure \ref{fig:Maguire}, black lines: each line corresponds to a different value of $x_0$) which cut contours of proportional population growth rate (Figure \ref{fig:Maguire}, grey lines) and thus separate equivalence classes.

\end{enumerate}

\section{Discussion}

I argued above that two ecosystems are equivalent if and only if for each species, the proportional population growth rate does not differ between the ecosystems. Much of community ecology is devoted to searching for patterns in variables such as species abundances or relative abundances, measures of diversity and measures of complexity. A consequence of my argument is that such activities will not lead to a deeper understanding of how organisms experience ecosystems. Instead, it may be more productive to search for patterns in proportional population growth rates.

Alternative stable states provide a good example of the consequences of this change of viewpoint. Determining whether a set of ecosystems has alternative stable states is a challenging problem \citep{Petraitis13}. However, if these states contain the same set of species, they are all equivalent to the organisms involved. The distinctive feature of alternative stable states is that the equivalence class in which all species have zero proportional population growth rate consists of more than one disjoint subset. Generalizing, it appears likely that when some species have nonzero proportional population growth rates, there will also exist equivalence classes consisting of more than one disjoint subset. For example, if the level sets of proportional population growth rates for two species are ellipsoids in niche space, then their intersection may consist of a pair of disjoint lower-dimensional ellipsoids. Roughly ellipsoidal level sets are expected if the axes of niche space include interactive resources \citep{Tilman80}, for which each species has a finite optimum value. This is not an outlandish biological situation.  Thus, equivalence classes consisting of disjoint sets of ecosystems with quite different physicochemical properties are plausible. This has been recognized, at least implicitly, in discussions of the ``double zero problem''. It has been claimed that shared absence of a species arising from an unsuitable environment does not provide useful information about ecosystem similarity because the environment may be unsuitable in different ways in different ecosystems \citep[e.g.][pp. 271-272]{Legendre12}. This is taking the view of an external observer who measures physicochemical conditions, rather than an organism in the ecosystem.

It is unlikely that two real ecosystems will ever be exactly equivalent. In empirical work, it may therefore be useful to measure how far two ecosystems are from being equivalent. This is analogous to the common approach of measuring dissimilarity between ecosystems \citep[chapter 7]{Legendre12}, but from the point of view of the organisms involved. What properties should be possessed by a measure of how far from equivalence two ecosystems $s_1, s_2$ are? Let $d(s_1, s_2)$ be such a measure. Convention suggests that we should have $d(s_1, s_2) \geq 0$ for all $s_1, s_2 \in S$. It will usually be sensible to require that $d(s_1, s_2) = 0$ if and only if $r(s_1) = r(s_2)$ (i.e. the two ecosystems are equivalent). There is in general no reason to privilege one ecosystem over another, so it is natural to require that $d(s_1, s_2) = d(s_2, s_1)$. A measure with all these properties is a semimetric \citep[p. 295]{Legendre12}. There are many measures with these properties, of which the most obvious is Euclidean distance, which also satisfies the triangle inequality, and is therefore a metric \citep[p. 39]{Sutherland09}: this last property may not be necessary, but is at worst harmless, and is often useful. Once a suitable measure has been chosen, popular activities such as ordination \citep[chapter 9]{Legendre12} can be done from the point of view of organisms, rather than that of an external observer. It may be useful to think of this measure as describing dissimilarity in niche space, as well as in growth space.

Given that the state $r(s)$ of an ecosystem (the $n$-tuple of proportional population growth rates for all the species it contains) is typically high-dimensional, it is natural to ask whether it can be summarized. Any function of $r(s)$ is invariant under dynamics within an equivalence class, and might therefore be considered as a summary of ecosystem state. In contrast, anything which is not a function of $r(s)$ will separate measures of the same equivalence class, and is therefore not a summary of ecosystem state. Many commonly-used measures of ``rate of succession'' such as Euclidean distances, Bray-Curtis distances and chi-square distances among relative abundances, reviewed in \citet[Appendix B]{Spencer15}, are not measures of ecosystem state, because they can take more than one value for the same value of $r(s)$, and therefore split up equivalence classes. As a first example of something that is a valid summary of ecosystem state, proportional population growth rates are likely to be unknown for most species in an ecosystem. In practice, it will be necessary to work with the $m$-tuple of proportional population growth rates that are known, where $m < n$. Since this is a function of $r(s)$, it is a summary of ecosystem state. As a second example, the Living Planet Index \citep{Loh05}, which is closely related to the mean of the elements of $r(s)$, and the shape change measure proposed by \citet{Spencer15}, which is the sample standard deviation of the elements of $r(s)$, are both scalar summaries of ecosystem state (although they were both originally presented as measures of change in ecosystem state). However, two ecosystems with the same value of one of these functions may not be equivalent from the point of view of any species. For example, two ecosystems may have the same mean of $r(s)$, even though the proportional population growth rates may differ between the ecosystems for every species. Thus, one could only claim equivalence from the point of view of a typical species, or a species chosen at random. This change in viewpoint from particular to aggregate properties demands a justification which is currently lacking, in the same way that studying biodiversity per se (an aggregate property) rather than the particular species in an ecosystem (each contributing to the aggregate property) demands a justification \citep[pp. 75-76]{Maier12}. There might be situations in conservation biology where the justification could be a ``veil of ignorance'' argument, similar to that used to agree principles of justice via the hypothetical situation in which individuals do not know their place in society \citep[p. 118]{Rawls99}. For example, an external observer might believe that species differ in their importance, even though the importances were unknown. In such a case, studying a typical species might be sensible. As far as I know, such arguments have not been developed in detail.

The structural difference between dynamics on equivalence classes of ecosystems and the dynamics of ecosystems has important consequences for the processes that generate repeatable patterns at the ecosystem level. An influential, if controversial, idea in ecosystem ecology is that ecosystems are shaped by natural selection on the ability to capture energy \citep{Lotka22}. Lotka's argument relies on the assumption that increased energy capture increases proportional population growth rate, and can therefore be subject to natural selection. Lotka proposed that such selection on energy capture leads to maximization of biomass and energy flow at the ecosystem level. However, variation within equivalence classes is invisible to natural selection, but visible to an external observer. Thus, natural selection cannot act on variation in ecosystem properties within equivalence classes. In the abstract language of \citet[p. 66]{Rosen78}, ecosystem dynamics are in general incompatible with the equivalence classes of ecosystems, and will thus appear acausal to an organism in the ecosystem. I argued above that equivalence classes may be disjoint, containing ecosystems with quite different physicochemical properties. If there are regularities at the ecosystem level, either they must be at the level of equivalence classes, or they must be generated by some mechanism other than natural selection. Stability selection is one such mechanism \citep{Borrelli15}. Stability selection acts ``without `seeing' the local environment'' \citep{Damuth18}, or in other words, it does not act via proportional population growth rates. It is therefore unlinked from the equivalence classes of an ecosystem, and may have the potential to generate regularities even within equivalence classes.

In conclusion, I distinguish between the view of ecosystems taken by an external observer, with the ability to study whatever they like, and an organism in an ecosystem, for whom only proportional population growth rates are visible. This distinction leads to major differences in the approach that should be taken to comparative studies of ecosystems, suggests open questions about how to summarize the state of an ecosystem, and clarifies the limits on the ability of natural selection to act on ecosystem structure.

\section*{Acknowledgements}

This work was funded by NERC grant NE/K00297X/1. I am grateful to Lev Ginzburg and Kevin Gross for thoughtful comments on an earlier version of this work.

\bibliographystyle{hapalike}
\bibliography{quotient}

\clearpage

\section*{Figure legends}

Figure \ref{fig:contours}. Movement through a two-dimensional niche space, with axes representing resources $y_1$, $y_2$. Grey lines: level sets of equal proportional population growth rate. Solid arrow: movement of the type considered by \citet{Maguire73} and \citet{Tilman80}. Dashed arrow: movement within a level set.

Figure \ref{fig:endomap}. Examples of endomaps $\alpha$ on a set of ecosystems $S$ for which $r$ is (a) or is not (b) a map in the category of sets with endomaps. In each case, the horizontal divisions in $S$ represent equivalence classes, with all points in a class mapped by $r$ to the same point in growth space $R$.

Figure \ref{fig:Maguire}. An ecosystem model in which the map $\alpha$ does not preserve equivalence classes. The $x$- and $y$-axes are concentrations of nitrogen ($N$) and phosphorus ($P$) in arbitrary units, and define a two-dimensional niche space. Grey lines are contours of constant $r$ (level sets in niche space). The bold black line is the contour $r = 1$. Thin black lines are some of the values to which the contour $r = 1$ is mapped after one unit of time by Equation \ref{eq:Maguire} (each line represents a different initial abundance $x_0$, between 0 to 5). Solutions obtained numerically. Parameter values: $r_{\max} = 10, N^* = 20, P^* = 0.2, a = 1, b = 1 \times 10^4, c = 6, d = 0.12$.

\clearpage

\begin{figure}[h]
  \includegraphics[width = \textwidth]{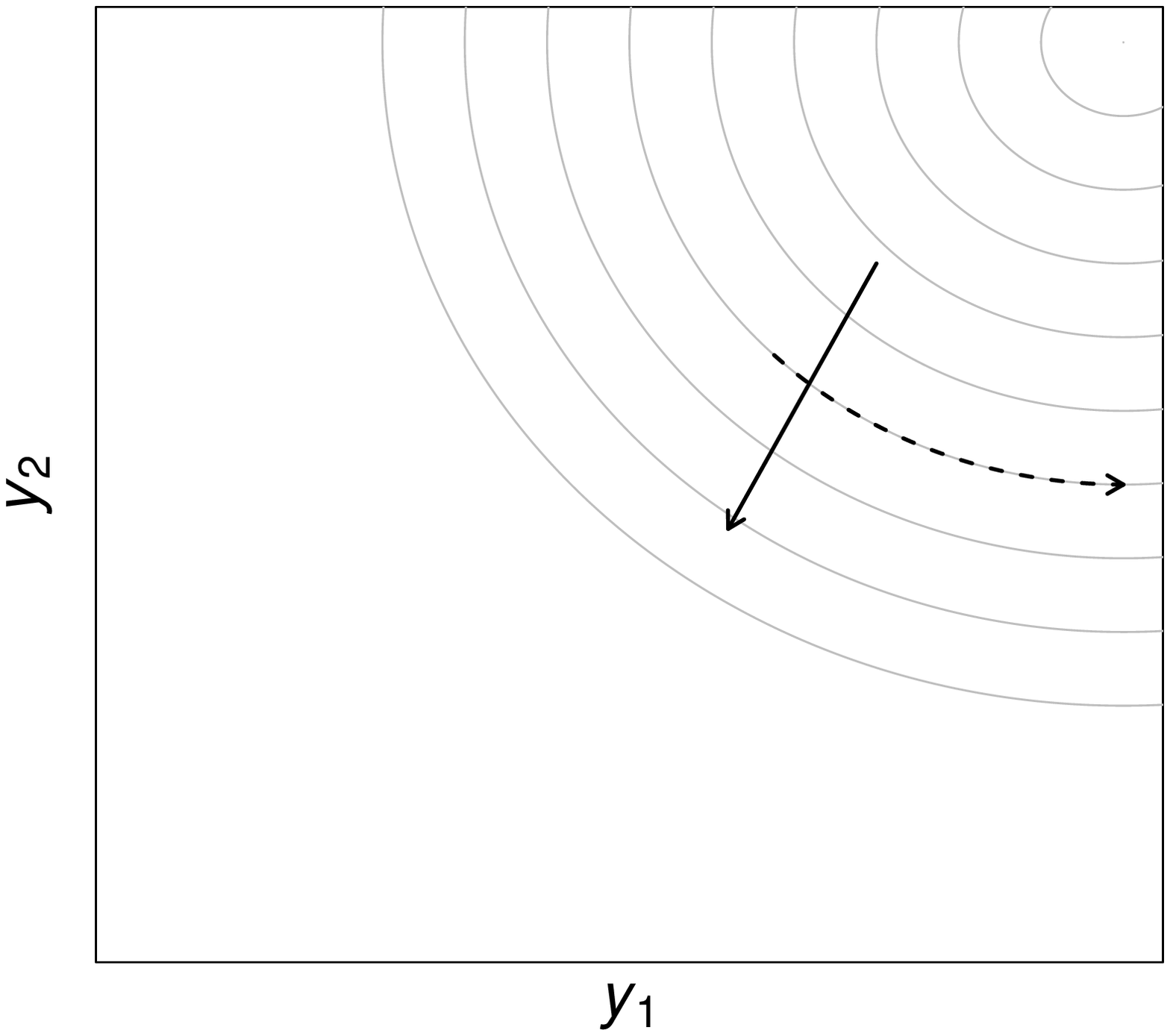}
  \caption{}
  \label{fig:contours}
\end{figure}

\begin{figure}[h]
\includegraphics[width = \textwidth]{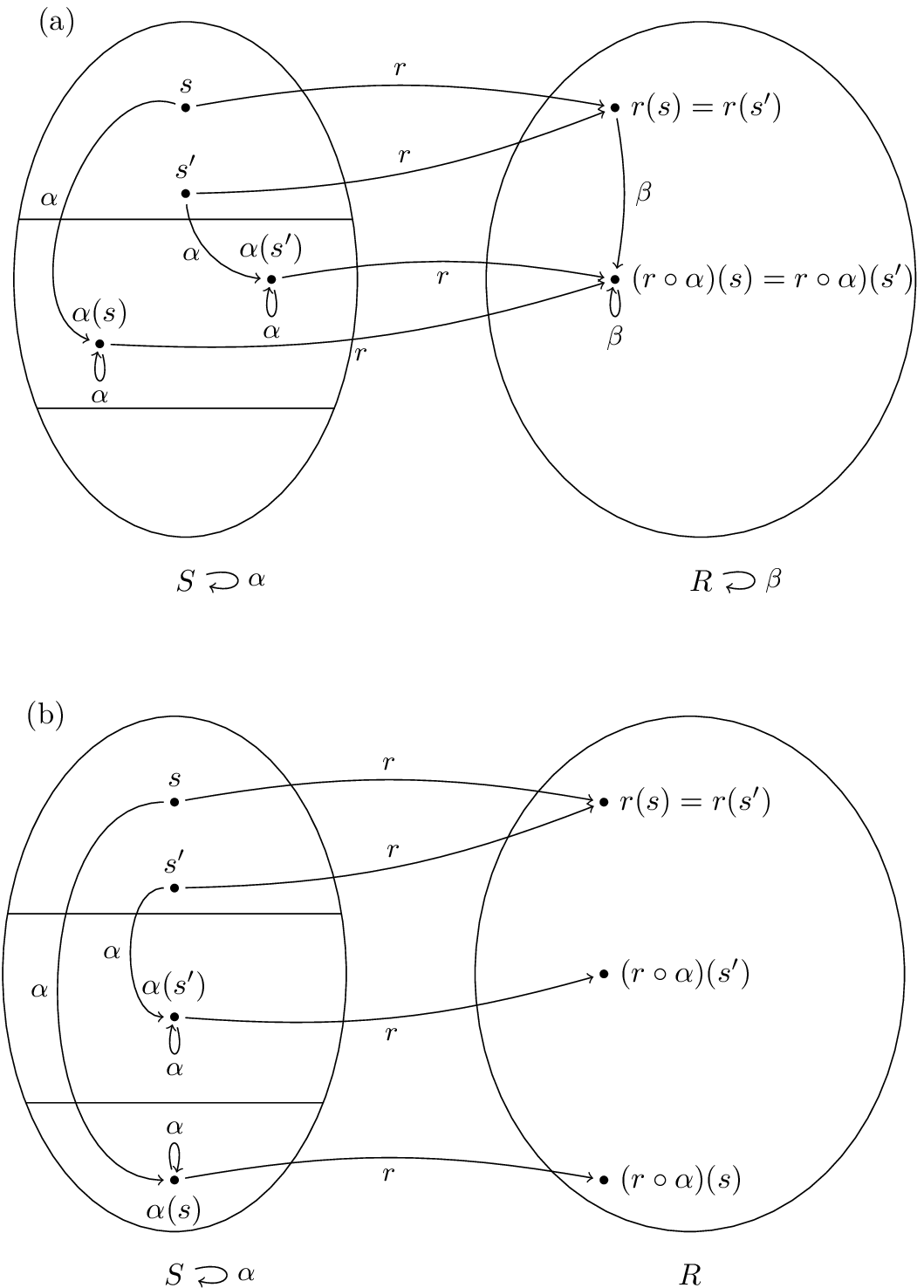}
\caption{}
\label{fig:endomap}
\end{figure}

\begin{figure}[h]
\includegraphics[width = \textwidth]{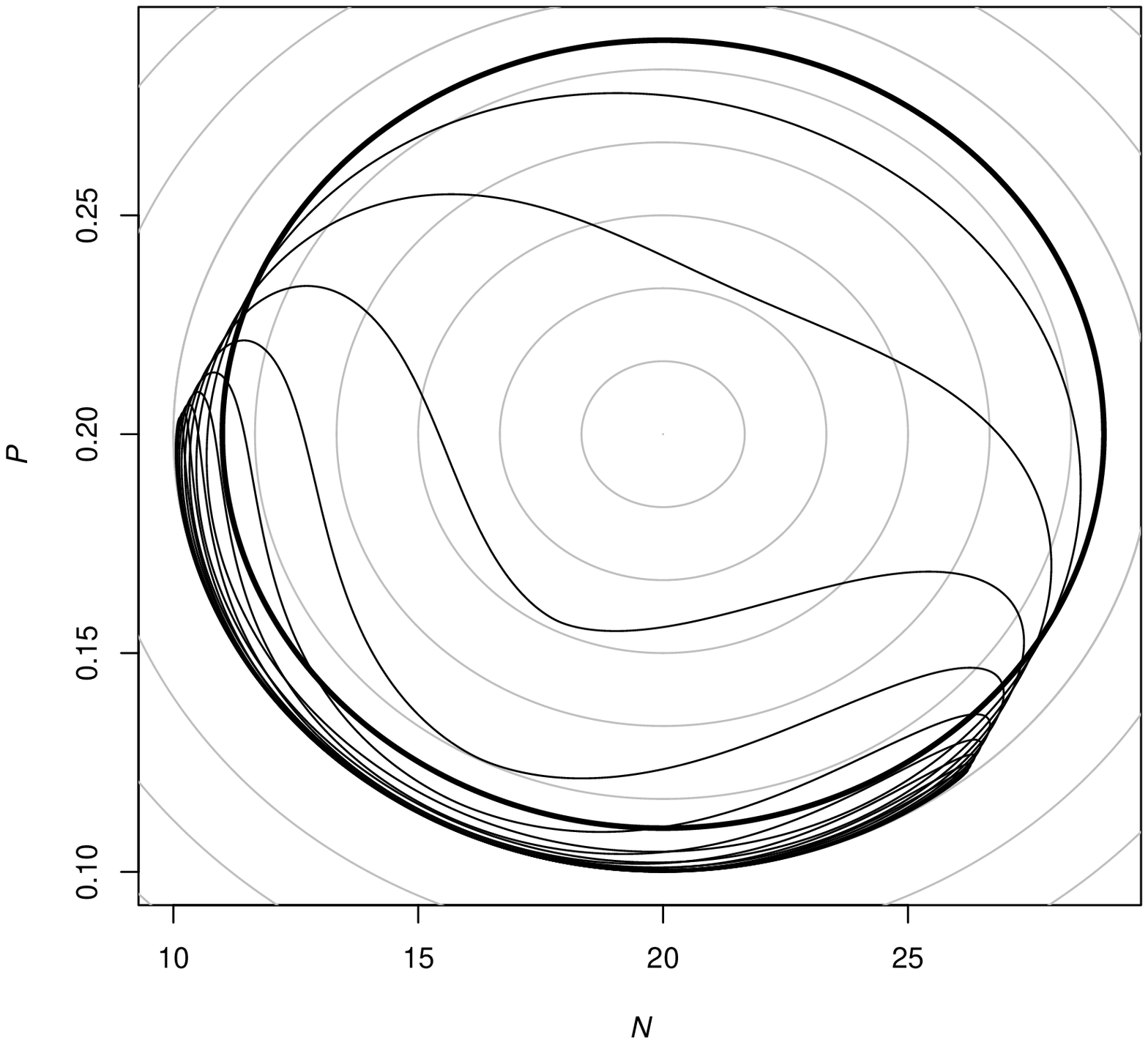}
\caption{}
\label{fig:Maguire}
\end{figure}

\end{document}